\documentclass[11pt]{article}
\usepackage[utf8]{inputenc}
\usepackage[T1]{fontenc}
\usepackage{moreverb}
\usepackage{amsmath,amssymb,bm}
\usepackage{graphicx}
\usepackage{url}
\usepackage{mathptmx}
\usepackage{algorithm}
\usepackage{algorithmic}
\usepackage{listings}
\newtheorem{algo}{Algorithm}[section]
\newtheorem{proof}{Proof}[section]
\newtheorem{prop}{Proposition}[section]

\newtheorem{remark}{Remark}[section]
\newtheorem{lemma}{Lemma}[section]

\usepackage{geometry}
\begin{document}
\begin{center}
{\bf Holonomic Decent Minimization Method for Restricted Maximum Likelihood Estimation}
\vspace*{5mm}\\
Rieko Sakurai$^{1}$ , and Toshio Sakata $^{2}$
\vspace*{5mm}\\
$^{1}$~{\it Graduate School of Medicine,  Kurume University 67 Asahimachi, Kurume 830-0011, JAPAN}\\
$^{2}$~{\it Faculty of Design Human Science, Kyushu University, 4-9-1 Shiobaru Minami-ku, Fukuoka 815-8540, JAPAN}\\
email: a213gm009s@std.kurume-u.ac.jp

\end{center}

\begin{abstract}

Recently, the school of Takemura and Takayama have developed a quite interesting minimization method called {\it holonomic gradient descent method} (HGD). It works by a mixed use of Pfaffian differential equation satisfied by an objective holonomic function and an iterative optimization method. They successfully applied the method to several maximum likelihood estimation (MLE) problems, which have been intractable in the past. On the other hand, in statistical models, it is not rare that parameters are constrained and therefore the MLE with constraints has been surely one of fundamental topics in statistics. In this paper we develop HGD with constraints for MLE . 
\end{abstract}

\section{Introduction}

Recently, the both schools of Takemura and Takayama have developed a quite interesting minimization method called holonomic gradient descent method(HGD). It utilizes Gr\"{o}bner basis in the ring of differential operator with rational coefficients. Gr\"{o}bner basis in the differential operators  plays a central role in deriving some differential equations called a Pfaffian system for optimization. HGD works by a mixed use of Pfaffian system and an iterative optimization method. It has been successfully applied to several maximum likelihood estimation (MLE) problems, which have been intractable in the past. For example, HGD  solve numerically the MLE problems for the von Mises-Fisher distribution and the Fisher-Bingham distribution on the sphere (see,  Sei et al.(2013)  and Nakayama et al.(2011)). Furthermore, the method has also been applied to the evaluation of the exact distribution function of the largest root of a Wishart matrix, and it is still rapidly expanding the area of applications(see, Hashiguchi et al.(2013)).
On the other hand, in statistical models, it is not rare that parameters are constrained and therefore the MLE problem with constraints has been surely one of fundamental topics in statistics. In this paper, we develop HGD for MLE problems with constraints, which we call the constrained holonomic gradient descent(CHGD).
The key of CHGD is to separate the process into (A) updating of new parameter values by Newton-Raphson method with penalty function and (B) solving a Pfaffian system.

\section{Constrained Optimization Problem}

We consider the following the constrained optimization problem.
\begin{eqnarray}
(P)&min&\ f(x) \nonumber \\
&s.t.& g_i(x) \leq 0,\ h_i(x)=0 
\end{eqnarray}
where $i=1,...,m$, $j=1,...,l$ 
and $f,g_i,h_j : R^n \rightarrow R$ are all assumed to be continuously differentiable function. $g_i(x)$ is an equality constraint function and $h_j(x)$ is an inequality constraint function. In this paper, the objective function $f(x)$ is assumed to be holonomic. 
We call the interior region defined by the constraint functions {\it the feasible region}. 

\subsection{Penalty Function Method}

A penalty function method replaces a constrained optimization problem by a series of unconstrained problems.
It is performed by adding a term to the objective function that consists of a penalty parameter $\rho$ and a measure of violation of the constraints. In our simulation, we use {\it the exact penalty function method}. The definition of the exact penalty function is given by (see Yabe (2006)).
\begin{eqnarray}
P(x; \rho):= f(x) + \rho\{\sum_{i=1}^m |g_i(x)| + \sum_{j=1}^l max(0,h_j(x))\}, \ \rho>0
\end{eqnarray}

\section{Holonomic descent method} 

Assume that we seek the minimum of a holonomic function $f(x)$ and the point $x$ which gives the minimum $f(x)$. In HGD, we use the iterative method together with a Pfaffian system. In this paper, we use the the Newton-Raphson iterative minimization method in which the renewal rule of the search point is 
given by
$$ 
x_{k+1}=x_{k}-H^{-1}(x_{k})\nabla f(x_{k}),  
$$  
where $\displaystyle{ \nabla f(x_{k})=(\frac{\partial f(x_{k})}{\partial x_{1}},\ldots, \frac{\partial f(x_{k})}{\partial x_{k}})^{T}}$ and $H(x_{k})$ is the Hessian of $f(x)$ at $x=x_{k}$.

\subsection{Mathematical background}

HGD is based on the theory of the Gr\"{o}bner basis.
In the following, we refer to the relation of a numerical method and the Gr\"{o}bner basis.

Let $R$ be the differential ring written as
\begin{eqnarray}
R = {\bm C}[x_1,..., x_n] \langle \partial_1,..,\partial_n \rangle \nonumber
\end{eqnarray}
where ${\bm C}[x_1,...,x_n]$ are the rational coefficients of differential operators.
Suppose that  $I=\{\ell_i|i=1,...,p\} $ is a left ideal of $R$, $k[\bm{x}]$ is a field and $D \in k[\bm{x}]\langle \partial_1,..,\partial_n \rangle \in I$. If an arbitrary function $f$ satisfies $Df=0$ for all $D$, then $f$ is a solution of $I$. That is 
\begin{eqnarray}\label{eq_h}
\ell_{i} f=0 \ \forall i
\end{eqnarray}
When $f$ satisfies Equation (\ref{eq_h}), $f$ is called {\it holonomic function}.

Let $s=[s_{1},\ldots,s_{t}]$,
with $s_{i}=q_{i}(x)\partial^{\alpha_{i}}$ be a standard basis in the quotient vector space $R/I$ which is a finite dimensional vector spaces. 
Let $G$ be the Gr\"{o}bner basis of $I$. The rank of arbitrary differential operations can be reduced by normalization by $G$. 
Assume that $\partial_{i} s_{j} \rightarrow_{G} \sum_{k}c^{i}_{jk}s_{k}$ holds. For a solution $f$ of $I$ put $F=(f,s_{2}f,\ldots,s_{t}f)^{T}$. Then, it holds that  
\begin{prop}(see, e.g.,Nakayama et al.(2011))
\begin{eqnarray}
\frac{\partial F}{\partial x_{i}}=P_{i}F,\ \ i=1,...,n
\end{eqnarray}
where $P_{i}$ is a $t \times t$ matrix with $c^{i}_{jk}$ as  a $(j,k)$ element  
\end{prop}
\begin{proof}
\begin{eqnarray}
\frac{\partial s_{j}f(x)}{\partial x_{i}} &=& (\partial_{i} \bullet s_{j})f(x)  \nonumber \\
                                          &=&  (\sum_{k}c^{i}_{jk}s_{jk})f(x) \ \  (mod \ \ I)  \nonumber \\
                                          &=&  [P_{i}F]_{j}, \ \ i=1,...,n,\ \ j=1...,t
\end{eqnarray}
This proves the assertion.
\end{proof}
The above differential equations are called {\it Pfaffian differential equations} or {\it Pfaffian system} of $I$. So we can calculate the gradient of $F$ by using Pfaffian differential equations. Then, $\nabla f(x_{k})$ and $H^{-1}(x_{k})$ are also given by Pfaffian differential equations. (see Hibi et al.(2012))
\begin{lemma}
Let $\sum_{j}^{t}a_{ij}s_{j}$ be the normal form of  $\partial_i=\partial / \partial x_i$ by $G$ 
and $\sum_{k}^{t}u_{ijk}s_{k}$ be the normal form of $\partial_i \partial_j$ by $G$. Then we have,
\begin{eqnarray}
\partial_{i} f(x_k) \rightarrow_{G} (\sum_{j}^{t}a_{ij}s_{j})f(x_k) =\sum_{j}^{t}a_{ij}F_j(x_k) = ((P_1F(x_k) )_1,...,(P_nF(x_k) )_1) \\
\partial_{i}\partial_{j}f(x_k) \rightarrow_{G} (\sum_{m}^{t}u_{ijk}s_{m})f(x_k)= \sum_{m}^{t}u_{ijm}F_m(x_k)=((\frac{\partial P_i}{\partial x_j}+P_iP_j)F(x_k))_1
\end{eqnarray}
where $(v)_1$ denotes the first entry of a vector $v$.
\end{lemma}

\subsection{Algorithm}

For HGD, we first give an ideal $I=\{\ell_i|i=1,...,p,\ \ell_{i} f=0 \ \forall i\} $ for holonomic function $f(x)$ and calculate the Gr\"{o}bner basis $G$ of $I$ and then the standard basis $S$ are given by $G$. The coefficient matrix $P_{i}$ for Pfaffian system is led by this standard basis, and  $H^{-1}(x_{k})$ and $\nabla f(x_{k})$ are calculated from $S$ by starting from a initial point $x_0$ through the Pfaffian equations. After these, we can compute automatically the optimum solution by a mixed use of  then Newton-Raphson method. The algorithm is given by below.
\newpage
\begin{algo}  
\end{algo}\begin{itemize}\setlength{\itemindent}{1.5mm}
\item[step 1]Set $k=0$ and take an initial point $x_{0}$ and evaluate $F(x_0)=(f(x_0),s_{1}f(x_0),\ldots,s_{t}f(x_0))^{T}$.
\item[step 2]Evaluate $\nabla f(x_{k})$ and $-H^{-1}(x_{k})$ from $F$ and calculate the Newton direction, $d_k=-H^{-1}(x_{k})\nabla f(x_{k})$
\item[step 3]Update a search point by $x_{k+1}=x_{k}+d_{k}$.
\item[step 4]Evaluate $F(x_{k+1})$ by solving Pfaffian equations numerically.
\item[step 5]Set $k=k+1$ and calculate $F(x_{k+1})$ and  goes to step.2 and repeat until convergence.
\end{itemize}

\begin{remark} The key step of the above algorithm is step 4. We can not  evaluate $F(x_{k+1})$ by inputting $x_{k+1}$ in the function $f(x)$ since the HGD treats the case that $f(x)$ is difficult to calculate numerically. 
Instead,  we only need calculate $f(x_0)$ and $F(x_0)$ numerically for a given initial value $x_0$.
\end{remark}

\section{Constrained holonomic gradient descent method}

Now, we propose the method in which we add constraint conditions to HGD and call it the constrained holonomic gradient descent method(CHGD).

\subsection{How to add the constraints}

For treating constraints  we use the penalty function and add it to objective function and make a new objective function and can treat it as the unconstrained optimization problem. We use HGD for evaluation of gradients and 
Hessian  and use the exact penalty function method for constraints. The value of updating a search point can be obtained as the product of directional vector and step size. The step size $\alpha$ is chosen so that the following Armijo condition is satisfied. In fact we chose $\alpha$ such that  
	\begin{eqnarray} \label{eq_s}
		P(x_k + \alpha x_k ; \rho) \leq P(x_k ; \rho) + \xi \alpha \{ P_l (x_k , \nabla x_k ; \rho) - P(x_k ; \rho) \},
	\end{eqnarray}
where $0 < \xi < 1$ and $P_l (x_k , \nabla x_k ; \rho)$ is the approximation of $P(x_k, \rho)$ given by.	
	\begin{eqnarray}
		P_l (x_k , \nabla x_k ; \rho) &= &f(x_k) + \nabla f(x_k)^T \nabla x_k \nonumber \\ 
		& + & \rho \{  \sum_{i=1} ^m | g_i (x_k) + \nabla g_i (x_k) ^T \nabla x_k | \nonumber  \\
		& + &  \sum_{j=i} ^n max(0, h_j(x_k) + \nabla h_j (x_k) ^T \nabla x_k )\}
	\end{eqnarray}
The initial value of $\alpha$ is set $1$ and then $\alpha$ is made smaller iteratively until $\alpha$ satisfies Equation (\ref{eq_s}), or $\alpha=0$. 

 In our algorithm, holonomic gradient descent plays a role to calculate the gradient vectors and then the penalty function plays a role to control the step size iteratively. 

\section{Computational results}

We apply CHGD for MLE for von Mises distribution(vM). The process of applying for HGD is shown in Nakayama et al.(2011).
The density function of vM is given by  $f(\kappa, \mu)=e^{\kappa \cos (\mu - x)}/\int_0^{2\pi} e^{\kappa \cos( \mu - t)}dt$. The parameters of vM, $\kappa$ and $\mu$, show concentration and mean of angle data ${\bm x}$ respectively. We set the parameters for MLE $\theta_1=\kappa \cos \mu$ and $\theta_2 = \kappa \sin \mu $. Now we solve the constrained optimization problem given by.
\begin{eqnarray}\label{optvM}
(P)&min&\ \ L(\theta_1,\theta_2) = e^{-\overline{c}\theta_1-\overline{s}\theta_2}\int_{0}^{2\pi} e^{\theta_1 \cos t+\theta_2 \sin t} dt \nonumber \\
&s.t.&\  \theta_1 \leq \theta_2
\end{eqnarray}
Let  ${\bm x}$ be sample data. Let $n$ be sample size. Then, $\overline{c}=\frac{1}{n} \sum_i^n \cos x_i$ and $\overline{s}=\frac{1}{n} \sum_i^n \sin x_i$. 

\subsection{Simulation}

In our simulation, we set the vM's parameter $(\kappa, \mu)=(5,\pi/4)$ of which the true value $(\theta_1, \theta_2)=(3.54, 3.54)$ and the initial value $(\theta_1, \theta_2)=(-2.0, 0.1)$. We tried the 2 patterns of constraints. Both of the case worked under the same condition except constraints. 
In Figure 1, the constraint is $\theta_1 \leq \theta_2$. In Figure 2, the constraint is  $\theta_1^2+\theta_2^2 \leq 9$.
Figures 1,2 are the drawing of the trace of the search point.
\newline
\begin{figure}[htbp]
\begin{tabular}{cc}
\begin{minipage}{0.5\hsize}
\begin{center}
  \includegraphics[bb=70 0 533 539, scale=.4]{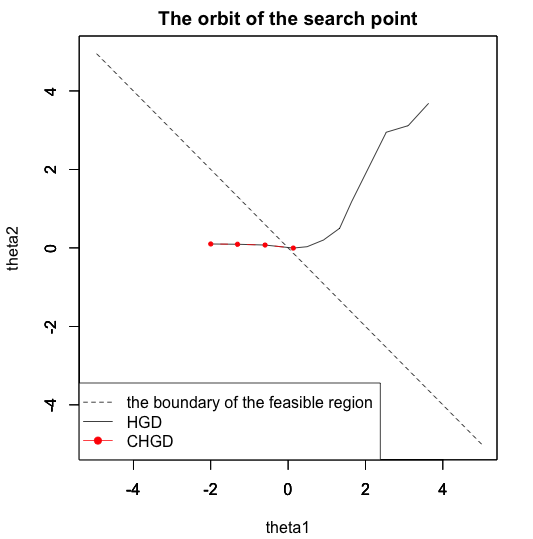}
 \caption{the case of $ \theta_1 \leq \theta_2$}\end{center}
\end{minipage}
\begin{minipage}{0.5\hsize}
\begin{center}
  \includegraphics[bb=70 0 533 539, scale=.4]{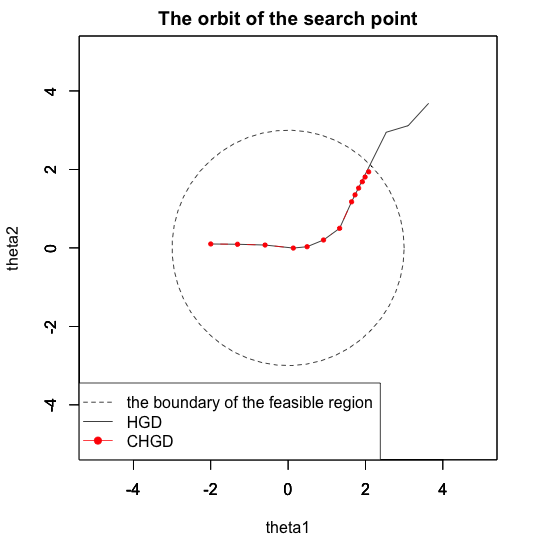}
 \caption{the case of $\theta_1^2+\theta_2^2 \leq 9$}
\end{center}
\end{minipage}
\end{tabular}
\end{figure}
\newline\newline
The result of simulation, the convergence point of HGD is $(\theta_1, \theta_2)=(3.63, 3.67)$. 
In Figure 1, the convergence point of CHGD is $(\theta_1, \theta_2)=(0.13, -0.004)$. In Figure 2, the convergence point of CHGD is $(\theta_1, \theta_2)=(2.08,  1.94)$. In the CHGD, the search direction is almost same as the HGD, because the direction is decided by the HGD's algorithm.
While, the constraints play the role to judge the search point is within the feasible region or not and decide the step size.
\newpage

\subsection{Runtimes}

CHGD is the effective method for optimization with constraints.
However, whenever CHGD increases the cost of runtimes than HGD regardless of whether the solution is in the feasible region or not.
The following table shows the runtimes when the optimization solution is within the feasible region. 
\newline
\begin{table}[htb]
  \begin{center}
  \caption{Comparing the runtimes\label{tb1}}
  \begin{tabular}{l||c|c|c} \hline
    & CPU TIME (sec) & \multicolumn{2}{c}{PATAMETERS $(\theta_1,\theta_2)$}   \\ \hline \hline
    HGD & 0.03698 &\ \ \ \ \ \ \ \ \ 2.120627\ \ \ \ \ \ \ \ \ \ & 2.120333 \\
    CHGD & 0.09834 & 2.120803 & 2.120629 \\
    NEWTON & 0.12598 & 2.124429 & 2.124855 \\\hline
    \multicolumn{4}{c}{\small We programmed by R and executed on Windows 7 64bit with RStudio Version 0.97.336}
  \end{tabular}
  \end{center}
\end{table}

In Table \ref{tb1}, all numbers are the means of 500 times trials. The optimization problem is Equation (\ref{optvM}). Sample data is drawn from the vM with $(\theta_1,\theta_2)=(2.12, 2.12)$. The third column of Table \ref{tb1} is the result with only Newton-Raphson method which optimize $f(x)$ directly, not use Pfaffian system. Thus, we see that HGD and CHGD is faster than Newton-Raphson method.

We see that the runtimes of CHGD is longer than HGD in general, where the both of solutions are almost the same value when the solution is inside  the feasible region. Sometimes the process finishes early by constraints, when the solution is outside the feasible region. Although, we need consider the cost of calculation of CHGD.


\begin{thebibliography}{99}
\bibitem{hashiguchi} Hashiguchi, H.,  Numata, Y.,  Takayama, N.,  Takemura, A. (2013). 
{\it "The holonomic gradient method for the distribution function of the largest root of a Wishart matrix"}. 
Journal of Multiva, riate Analysis 117 (2031) 296-312

\bibitem{sei} Sei, T., Shibata, H., Takemura, A., Ohara, K., Takayama, N. (2013). 
{\it "Properties and applications of Fisher distribution on the rotation group"}. 
Journal of Multivariate Analysis.

\bibitem{dojo} Hibi, T., Hamada, T., Noro, M., Aoki, S., Takemura, A., Osugi, H., Takayama, N., Nakayama, H., Nishiyama, K. (2012). 
{\it "Gr\"{o}ebner  dojo(japanese)"}. Kyoritsu publisher.

\bibitem{nakayama} Nakayama, H., Nishiyama, K., Noro, M., Ohara, K., Sei, T., Takayama, N., Takemura, A. (2011). 
{\it "Holonomic gradient descent and its application to the Fisher–Bingham integral"}.
 Advances in Applied Mathematics, 47(3), 639-658.
\bibitem{yabe}  Yabe, H. (2006). 
{\it "Introduction and Application of Optimization Problem(japanese)"}.
 Surikougakusha publisher.
\bibitem{d.cox} Cox, D. A., Little, J., O'Shea, D. (2007). 
{\it "Ideals, varieties, and algorithms: an introduction to computational algebraic geometry and commutative algebra (Vol. 10)"}.
 Springer Verlag.
\end{thebibliography}
\end{document}